\documentclass[showpacs,amsmath,amssymb,twocolumn,10pt,pra,superscriptaddress]{revtex4-1}

\usepackage[dvips]{graphicx} 
\usepackage{amsfonts}
\usepackage{amssymb}
\usepackage{amscd}
\usepackage{amsmath}    
\usepackage{enumerate}
\usepackage{epsfig}
\usepackage{subfigure}
\usepackage{xcolor}
\usepackage{amsthm}
\usepackage{framed}
\usepackage{multirow}
\usepackage{epstopdf}

\usepackage[utf8]{inputenc}

\usepackage{color}
\usepackage{hyperref}
\hypersetup{  colorlinks=true, linkcolor=blue, citecolor=red, urlcolor=blue  }

\usepackage{physics}

\newtheorem{theorem}{Theorem}
\newtheorem{lemma}{Lemma}

\newtheorem{definition}{Definition}

\newcommand{\comments}[1]{}

\begin{document}
\preprint{APS/123-QED}
\title{Interaction-free measurement study as a quantum channel discrimination problem}

\date{\today}

\author{You Zhou}
\affiliation{Center for Quantum Information, Institute for Interdisciplinary Information Sciences, Tsinghua University, Beijing 100084, China}
\author{Man-Hong Yung}
\email{yung@sustc.edu.cn}
\affiliation{Institute for Quantum Science and Engineering and Department of Physics,
South University of Science and Technology of China, Shenzhen 518055, China}

\begin{abstract}
Interaction-free measurement (IFM), just as its name implies, can enable one to detect an object without interacting with it, i.e., substantially reducing the damage to the object. With the help of quantum channel theory , we investigate the general model of ``quantum-Zeno-like" IFM, whose optics setup is a Mach-Zehnder like interferometer utilizing the quantum Zeno effect, where the object to be detected is semitransparent and the interrogation cycle number $N$ is finite. And we define two important probabilities $P_{\rm loss}$ and $P_{\rm error}$ to evaluate the IFM process, which describe the photon loss rate and the error of discriminating the presenece/absence of the object respectively. The minimum values of these two probabilities and the corresponding initial input states to reach them are attained via this model. And we find that when the interrogation cycle $N$ approaches infinity, the object can be perfectly detected, where the minimum values of these two probabilities are both zero and the initial input state to reach them becomes the same state $|1\rangle$ in our system. In addition, we also study whether quantum correlation can benefit IFM or not, but the answer is no, in the sense that the entangled photon input state cannot minimize $P_{\rm loss}$,  $P_{\rm error}$ more than single photon input state. Our work provides principal theoretic support for the practical realization of IFM and the employed analysis technique can be applied to other quantum facilitating scenarios.
\end{abstract}

\maketitle

\section{introduction}
Quantum measurement is a fundamental concept in the quantum theory. The measurement process extracts the information stored in the quantum system to the classical world, where the quantum state is required to change adaptively based on the measurement outcome.
In fact, the measurement on the target system, say $A$, is accomplished indirectly by coupling it to another system $B$ and implementing measurement on that system alternatively.
However, even the nonobservance of a particular result of $B$ would modify the wave function of $A$, which is called the ``negative result measurement"~\cite{renninger1960messungen,dicke1981interaction}.

Following these former works,  Elitzur and Vaidman introduced a ``counterfactual" protocol dubbed interaction-free measurement (IFM) \cite{elitzur1993quantum}. In this IFM protocol, a photon is sent to a standard Mach-Zehnder interferometer to detect an opaque object, where the maximum efficiency for a successful detection without photon absorption is $50\%$ \cite{elitzur1993quantum,kwiat1995interaction}. However, by a modification on account of the quantum Zeno effect \cite{misra1977zeno}, the efficiency can approach $100\%$ as the interrogation cycle goes to infinity~\cite{kwiat1995interaction,kwiat1999high}.

Interaction-free measurement has been used to detect fragile objects, like single atom \cite{karlsson1998interaction,volz2011measurement} or photon-sensitive substances \cite{inoue2000experimental}. And the application to electron microscopy is also developed \cite{putnam2009noninvasive,thomas2014semitransparency,Kruit201631}, which should facilitate the biological molecules imaging.

Besides the original optical setup \cite{kwiat1995interaction,kwiat1999high}, there are many other different schemes proposed \cite{paraoanu2006interaction,putnam2009noninvasive,chirolli2010electronic,PhysRevB.93.115411}, or realized
\cite{HAFNER1997563,tsegaye1998efficient,ma2014chip} to achieve ``quantum-Zeno-like" IFM. However, the physical model behind them is essentially the same; they all involve utilizing the quantum Zeno effect to keep the photon state unchanged, in the presence of an object.

Here we consider an optical setup \cite{kwiat1999high} to illustrate the principle of IFM, as showed in Fig.~\ref{optical_setup}.  Let us denote respectively the state $|1\rangle$, $|2\rangle$, and $|3\rangle$ as the representation for,
\begin{equation}\label{state_label}
{\rm up}  \Leftrightarrow |1\rangle, \quad {\rm down}  \Leftrightarrow |2\rangle, \quad {\rm loss} \Leftrightarrow |3\rangle
\end{equation}
state of the incident photon. A light-absorbing object (e.g. a photon-sensitive bomb in \cite{elitzur1993quantum}) is placed in the path of the down state photon.  And the probability for this object to appear is denoted by $\Pr({\rm here})=q$.

In fact, in order to describe the incident photon state transformation when the object is present explicitly, we mimic the effect of this object with a mirror, followed by a photon detector \cite{kwiat1995interaction} (see Fig.~\ref{optical_setup} for detailed  illustration).

\subsection{Interaction-free measurement}
Let us consider the essential idea of interaction-free measurement: first, an incident photon is prepared in the up path, with a quantum state labeled by $\ket{1}$. Then, the incident photon is rotated by an angle~$\theta$ through a beamsplitter,
\begin{equation}\label{rotation_theta}
{R_\theta }\left| 1 \right\rangle  = \cos \theta \left| 1 \right\rangle  + \sin \theta \left| 2 \right\rangle,
\end{equation}
where $\theta=\pi/{2N}$. Here $N$ will be identified as the total number of interrogation cycle.

{\bf Presence of the object:}
If there exists an object along the down path, the photon in the down state $|2\rangle$ will be totally transferred to the loss state $|3\rangle$ by the mirror, i.e.,
\begin{equation}\label{perfect_abs}
{U_{{\text{I}}}}\left| 2 \right\rangle  = \left| 3 \right\rangle,
\end{equation}
where the subscript of $U_{\text{I}}$ stands for interaction. Furthermore, when applying it to a quantum superposition, we have
\begin{equation}\label{eq_qusuper_13}
{U_{{\text{I}}}}(\cos \theta \left| 1 \right\rangle  + \sin \theta \left|2 \right\rangle ) = \cos \theta \left| 1 \right\rangle  + \sin \theta \left| 3 \right\rangle.
\end{equation}

Then followed by the projective measurement $M$ by the photon detector,
\begin{equation}
\begin{aligned}
&P_0=\ket{1}\bra{1}+\ket{2}\bra{2},\\
&P_1=\ket{3}\bra{3},\\
\end{aligned}
\end{equation}
the final state in Eq.~(\ref{eq_qusuper_13}) becomes a mixed state,
\begin{equation}
\cos^2 \theta \ket{1}\bra{1} + \sin^2 \theta \ket{3}\bra{3} ,
\end{equation}
where the probability of the photon traveling along up path without absorption is given by $\Pr(\ket{1})=\cos^2\theta$. And the probability of which the photon is transformed to the loss state $\ket{3}$ and absorbed by the detector is $\Pr(\ket{3})=\sin^2\theta$.

In the probability subspace of $P_1$, the loss state photon will not participate in the following interrogation cycle, i.e., the IFM process halts in this case.
Consequently, the probability of finding $|1\rangle$ after $N$ cycles equals to $\Pr(\ket{1})=\cos^{2N}(\theta)$.  When $N$ approaches infinity, we have
\begin{equation}
\lim_{N\rightarrow\infty}\cos^{2N}(\theta)=1.
\end{equation}
Therefore, one can find the final state to be $|1\rangle$ with probability $1$ without any photon loss, in the presence of an object.

{\bf Absence of object:}
If there is no object, i.e., the down state $\ket{2}$ will travel straight through without getting absorbed (or reflected by the mirror); the rotation $R_\theta$ is directly applied $N$ times. Thus, the input photon state $|1\rangle$ can be rotated to $|2\rangle$ at the end, that is,
\begin{equation}
{({R_\theta })^N}\left| 1 \right\rangle  = {R_{N\theta }}\left| 1 \right\rangle  = \left| 2 \right\rangle.
\end{equation}

In summary, after $N$ cycles, if we get $\ket{1}$, it implies the existence of the object, while $\ket{2}$ implies the absence; we can therefore unambiguously detect the presence of an object (because state $|1\rangle$ and $|2\rangle$ are orthogonal), without any photon absorption by the object. This is the essential idea of the interaction-free measurement, based on the physics of quantum Zeno effect.

\begin{figure}[t]
\centering
\resizebox{6.5cm}{!}{\includegraphics[scale=0.8]{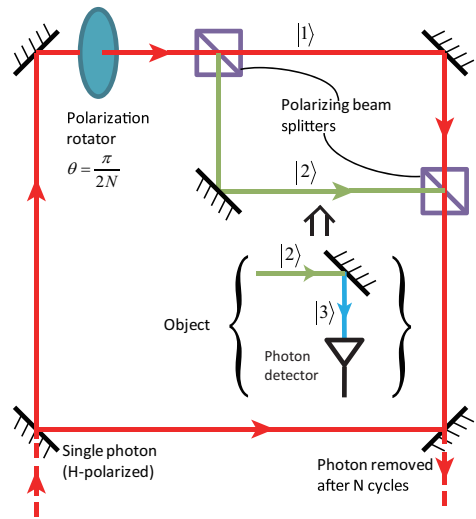}}
\caption{The ``quantum-Zeno-like" IFM setup. We illustrate the principle of IFM using the optical scheme in ref \cite{kwiat1999high}. The polarization rotator can rotate the photon polarization by $\theta=\frac{\pi}{2N}$ in each cycle. And the polarizing beam splitter can separate the photon to up or down path if the photon is horizontal polarized $H$ or vertical polarized $V$. So the polarizations of the photon label the up state $|1\rangle$ and the down state $|2\rangle$ in $Eq.~(\ref{state_label})$ respectively. In addition, the object is mimicked by a mirror followed by a photon detector. The mirror can transform the down state $\ket{2}$ to the loss state $\ket{3}$. And the photon detector implements the projective measurement on the two $\{\ket{1},\ket{2}\}$ and $\{\ket{3}\}$ subspaces. After $N$ interrogation cycles, we can judge whether there is an object in the down path with the final polarization state of the photon, without any photon absorption.}\label{optical_setup}
\end{figure}

\subsection{Finite rounds and imperfect absorption}

In practice, there are two problems one should consider, in implementing the interaction-free measurement. First, the number of interrogation cycle $N$ has to be finite; it is also impossible to make the rotation angle arbitrarily small.

Second, the absorption of photon by object may not be perfect, as assumed in Eq.~\eqref{perfect_abs}. In general, we should consider the absorption probability to be less than unity, i.e.,
\begin{equation}
{U_{{\text{I}}}} \left| 2 \right\rangle  =  a\left| 2 \right\rangle  + \sqrt {1 - {a^2}} \left| 3 \right\rangle ,
\end{equation}
where $a^2$ characterizes the transparency of the object. Here $a$ is assumed to be an non-negative real number for simplicity.

In this scenario, we can substitute a beam splitter, whose transparency is  $a^2$, for the mirror in Fig.~\ref{optical_setup} to mimic the corresponding semitransparent object. This treatment is similar to \cite{garcia2005quantum}, and other works \cite{thomas2014semitransparency,mitchison2001absorption} gave different but equivalent treatments.

\subsection{Related works}
Previous work has shown that the successful rate of IFM decreases if the object is semitransparent, compared with the opaque case~\cite{jang1999optical,vaidman2003meaning,garcia2005quantum}. The performance can be improved by increasing the interrogation cycle number $N$ and the object can also be detected perfectly without any photon absorption when $N\rightarrow\infty$ \cite{kwiat1998experimental,azuma2006interaction}.

In the literature ~\cite{garcia2005quantum,kwiat1998experimental,jang1999optical,azuma2006interaction}, the initial input state is usually taken as a pure state, namely $\ket{1}$. In the presence of an object, the successful probability
\begin{equation}
P_{\rm success}=|\bra{1}\hat{O}_{\rm IFM}\ket{1}|^2 \ ,
\end{equation}
is used to characterize the performance of the IFM process. Here $\hat{O}_{IFM}$ is a linear operator, but not necessarily unitary due to the possibility of photon loss. The value of $P_{suc}$ specifies the probability that one can receive a $\ket{1}$ photon after sending a $\ket{1}$ photon at the beginning, in the presence of an object. In this case, one can confirm the presence of an object without photon being absorbed.

To the best of our knowledge, there is no work aiming to optimize the IFM process through a search of optimal input states of the photon. In particular, the possibility of using quantum correlation to enhance the ability of channel discrimination have been achieved in the context of quantum illumination \cite{lloyd2008enhanced}, and here we also study the possibility of this kind of enhancement in IFM process.

\subsection{Main results}
In this work, we provide analytic solutions to a generalized IFM model. To be specific, we focus on two main quantities to benchmark the performance of IFM, namely (i) the loss probability $P_{\rm loss}$ and (ii) the error probability $P_{\rm error}$,  which respectively describe the photon loss rate and the minimum error of discriminating the object. Specifically, the minimum values of these two probabilities are investigated analytically, for any given values of the object transparency~$a^2$ and the interrogation number $N$.

Our main results are summarized as follows:
\begin{itemize}

\item[ \ ] {\bf For unentangled input states}:

\item There exists a unique quantum state $\ket{\varphi_0}$ minimizing $P_{\rm loss}$, for any finite $N$, which approaches~$0$ asymptotically as $N\rightarrow\infty$.

\item There are two states $\ket{\varphi_\pm}$ that leads to $P_{\rm error}=0$, i.e. perfect discrimination, as long as the following inequality is fulfilled,
\begin{equation}
\frac{1+a}{1-a} \ \sin(\frac{\pi}{2N})\leq1.
\end{equation}

\item The photon loss rate of $\ket{\varphi_+}$ is smaller than that of $\ket{\varphi_-}$, i.e., $(P_{\rm loss})_{\ket{\varphi_+}}<(P_{\rm loss})_{\ket{\varphi_-}}$, which means $\ket{\varphi_+}$ is better than $\ket{\varphi_-}$ in term of $P_{\rm loss}$.

\item For $N\rightarrow\infty$, both $\ket{\varphi_0}$ and $\ket{\varphi_+}$ approach the same state $\ket{1}$, where both $(P_{\rm loss})_{\ket{\varphi_0}}$, $(P_{\rm loss})_{\ket{\varphi_+}}$ share similar  asymptotic behavior $O(1/N)$.

\end{itemize}

In addition, we studied how quantum correlation of input states can facilitate the IFM process by utilizing entangled photons in the setting of quantum illumination \cite{lloyd2008enhanced}: send one photon in an entangled pair to the IFM cycle but keep the other photon. At the end, a joint POVM measurement is performed on both photons. We found that
\begin{itemize}
\item [ \ ] {\bf For entangled input states:}

\item The optimal state to reach the minimal $P_{\rm loss}$ is the product state $\ket{\varphi_0}|\phi_0\rangle$, where
$|\phi_0\rangle$ is any state of the second photon.

\item The two solutions $\ket{\varphi_\pm}$ expand to a family of quantum states in the larger Hilbert space. Specifically, all members of the form,
\begin{equation}
\alpha\ket{\varphi_+}|\phi_1\rangle+\beta\ket{\varphi_-}|\phi_1^\perp\rangle,
\end{equation}
can be employed to achieve $P_{\rm error}=0$, where $|\phi_1\rangle, |\phi_1^\perp\rangle$ are any two orthogonal states of the second photon.
However, the one with the minimal $P_{\rm loss}$ in this family of states is the unentangled state $\ket{\varphi_+}|\phi_1\rangle$.
\end{itemize}

In other words, entangled photons cannot minimize $P_{\rm loss}$,  or $P_{\rm error}$ better than the case with unentangled photons. Therefore, we conclude that entanglement cannot improve the IFM process.

The rest of the paper is organized as follows. In Section \ref{sec2}, we construct a general model with the use of quantum channel. In Section \ref{sec3}, we simplify the quantum channels for pure input state. In Section \ref{sec4}, \ref{sec5}, we study the case with opaque object and semitransparent object respectively. We conclude in Section \ref{sec6}.

\section{the general model}\label{sec2}
In this section, we present a general model of interaction-free measurement, taking into account of a semitransparent object and a finite number of interrogation cycle. In addition, we shall consider sending entangled photons as the input state as well.

First of all, the IFM process can be described as a quantum channel, which is sequentially-applied $N$ times on the input photon state, depending on the presence/absence of the object.  Thus, detecting the object is equivalent to a channel discrimination problem.

In both cases, a unitary rotation operator (see Eq.~(\ref{rotation_theta}))
\begin{equation}\label{rota_3state}
  R_\theta=\left(
  \begin{array}{ccc}
    \cos\theta & -\sin\theta & 0\\
    \sin\theta & \cos\theta & 0\\
    0 &  0 & 1
  \end{array}
\right)
\end{equation}
is applied for each step at first, where the matrix is wriiten in the $\ket{1}, \ket{2}$ and $\ket{3}$ basis. It can be regarded as the following channel,
\begin{equation}
\mathcal{E}_\theta(\rho)=R_\theta \, \rho \, R_\theta^\dag,
\end{equation}
where $\rho$ is the density matrix of the input state.

If a generic semitransparent object is present, the partial absorption effect can be represented by an effective quantum channel $\mathcal{E}_I$ ($I$ is short for interaction) on the photon state (see Appendix.~\ref{channel_derivation} for detailed derivation):
\begin{equation}\label{}
 \mathcal{E}_I(\rho)=\sum_{i=0,1} A_i\rho A_i^\dag,
 \end{equation}
\begin{equation}\label{channel_123}
\begin{aligned}
  A_0=&|1\rangle \langle1|+a|2\rangle \langle2|+|3\rangle \langle 3|, \\
  A_1=&\sqrt{1-a^2}|3\rangle \langle2|,
\end{aligned}
\end{equation}
where $A_0$ and $A_1$ are the Kraus operators fulfilling $\sum_{i=0,1}A_i^\dag A_i=I$.  $A_0$ describes the process that the down state $\ket{2}$ component partially decays to the loss state $\ket{3}$ component; $A_1$ is for the increase of the population on loss state, which indicates the photon loss probability.

It is necessary to clarify that $\mathcal{E}_I$ is not just a mere combination of the unitary $U_I$ and the projective measurement $M$, since the population on loss state $\ket{3}$ component will be absorbed by the detector and not participate in the following cycle, indicating that the photon loss is an irreversible process.

Then the channels that describe the whole interrogation can be written down by cycling the above channels for $N$ times as below
\begin{equation}
 \rho'=[\mathcal{E}_\text{I} \ \mathcal{E}_\theta]^N(\rho)=\mathcal{E}'(\rho),
 \end{equation}
\begin{equation}
 \rho''=[\mathcal{E}_\theta]^N(\rho)=\mathcal{E}''(\rho),
 \end{equation}
where $\rho'$ is the output density matrix, if the object is present; $\rho''$ is the output density matrix for the object absence case. The corresponding overall quantum channels are denoted by $\mathcal{E}'$ and $\mathcal{E}''$ respectively.

Our main concerns in IFM are two probabilities: one is the photon loss probability $P_{\rm loss}$, which describes the damage to the object. This concern is important if the detected object is fragile, like electronic devices or biological matters.  In fact, it is just the population accumulating on the loss state $\ket{3}$ component after the full IFM process,

\begin{equation}\label{def_Ploss}
  P_{\rm loss}=q  \ \langle 3|\mathcal{E}'(\rho)|3 \rangle,
 \end{equation}
where $q$ is the probability for the presence of the object.

\begin{figure}[t]
\centering
\resizebox{8cm}{!}{\includegraphics[scale=0.8]{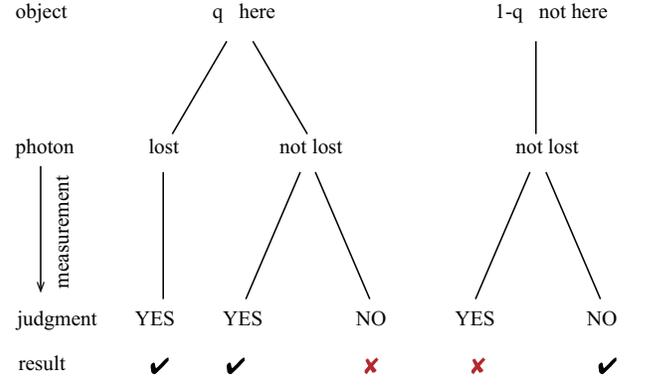}}
\caption{Illustration for the error happening in IFM. At the end of IFM, we implement POVM on the final output photon state, and we will make error when giving the judgment NO (YES) in the presence (absence) of the object.
The final line shows the result (right or wrong) of the judgment. And we can find that if the photon is lost, one can definitely confirm the presence of the object since the photon should not be lost in the absence of the object. Hence we can give the right judgment YES and make no error in this case.}\label{fig_yesno}\end{figure}

The other is the probability for making a error in the detection. Given the object existing probability $\Pr(\rm here)$, the error happens when one give the wrong judgement after the interrogation cycles (see Fig.~\ref{fig_yesno}), that is,
\begin{equation}\label{error}
\begin{aligned}
  P_{\rm error}=&\Pr(\rm here)\Pr(\rm NO|\rm here)\\
  &+\Pr(\rm not\ here)\Pr(\rm YES|\rm not\ here),
\end{aligned}
\end{equation}
where one gives the judgment NO in the presence of the object or YES in the absence of the object.

To be specific, one sends a input photon state $\rho$, and receives the output photon state $\mathcal{E}'(\rho)/\mathcal{E}''(\rho)$ depending on the presence/absence of the object. Then one makes the judgment by implementing a two-value POVM measurement $\{\Pi_1,\Pi_2\}$ on the final output photon. Here $ \Pi_1,\Pi_2$ are positive operators fulfilling $\Pi_1+\Pi_2=I$.  Specifically, if obtaining the measurement result $1(2)$, one makes the judgment that the object is here(not here). Hence, the corresponding conditional probabilities in the above equation become
$\Pr(\rm NO|\rm here)=\Tr[\Pi_2\mathcal{E}'(\rho)]$ and
$\Pr(\rm YES|\rm not\ here)=\Tr[\Pi_1\mathcal{E}''(\rho)]$. And substitute them into Eq.~\eqref{error}, we get
\begin{equation}\label{error}
\begin{aligned}
P_{\rm error}=q\Tr[\Pi_2\mathcal{E}'(\rho) ]+(1-q)\Tr[\Pi_1\mathcal{E}''(\rho)].
\end{aligned}
\end{equation}
where the definition $\Pr(\rm here)=q$ is used.

Following the minimum-error scheme \cite{helstrom1976quantum,herzog2004minimum} in quantum state discrimination, by choosing a proper POVM measurement, the reachable minimal error shows the following form,
\begin{equation}\label{def_Perror}
  P_{\rm error}=\frac{1}{2}[1-\|q\mathcal{E}'(\rho)-(1-q)\mathcal{E}''(\rho)\|],
 \end{equation}
where $||O||=Tr(\sqrt{O^\dag O})$ denotes the trace norm of any operator $O$. It indicates that the lager the trace norm distance between the two output state $\mathcal{E}'(\rho)/\mathcal{E}''(\rho)$ normalized by the corresponding probabilities $q/1-q$, the smaller the error is.

One may argue that in each cycle the photon detector may click (bomb exploding in \cite{elitzur1993quantum}), then the presence of the object can be confirmed at the middle of the whole process, thus it is not necessary to finish the following interrogation and discriminate the state at the end. However, in fact, they are equivalent; as will be showed explicitly in Eq.~\eqref{Perror_pure}, the loss probability of every cycle that accumulates on the loss state $\ket{3}$ component can be excluded from the $P_{\rm error}$, just because we can always make no error and confirm there is an object if the photon is lost (see Fig.~\ref{fig_yesno}).

The main focus of our IFM study is to find the minimums of these two probabilities $P_{\rm loss}$ and $P_{\rm error}$, and the initial input photon states to reach them. Fortunately, with the following theorem, we can reduce the range of the input state from any density matrix $\rho$, say mixed or pure, to just pure state $\ket{\varphi}$ in the Hilbert space of the photon.

\begin{theorem}\label{t1}
The minimums of the loss probability $P_{\rm loss}$ and the error probability $P_{\rm error}$ can be both reached by the pure state.
\end{theorem}

\begin{proof}
Due to the linearity of the quantum channel, we have
\begin{equation}\label{mix_pure_loss}
\begin{aligned}
 P_{\rm loss}&=q\langle 3|\mathcal{E}'(\rho)|3 \rangle,\\
&=q\langle 3|\mathcal{E}'(\sum_i p_i\varphi_i)|3 \rangle,\\
&=q\langle 3|\sum_i p_i\mathcal{E}'(\varphi_i)|3 \rangle,\\
&=\sum_i p_i q \langle 3|\mathcal{E}'(\varphi_i)|3 \rangle,\\
&=\sum_i p_i P_{\rm loss}^i,
\end{aligned}
\end{equation}
where $\varphi_i$ represents the density matrix of the pure state $|\varphi_i\rangle$, $P_{\rm loss}^i$ is the corresponding loss probability for it, and $\sum_i p_i\varphi_i$ is any convex decomposition of the input state $\rho$.

Eq.~(\ref{mix_pure_loss}) shows that the loss probability $P_{\rm loss}$ of the mixed state $\rho$ equals to the weighted average of $P_{\rm loss}^i$ of the corresponding pure state.  Thus there is at least one pure state $\varphi_i$ whose loss probability $P_{\rm loss}^i\leq P_{\rm loss}$.

Moreover, combining  the convex property of the trace norm, we also have
\begin{equation}
\begin{aligned}
  P_{\rm error}&=\frac{1}{2}[1-\|q\mathcal{E}'(\rho)-(1-q)\mathcal{E}''(\rho)\|],\\
 &=\frac{1}{2}[1-\|q\mathcal{E}'(\sum_i p_i\varphi_i)-(1-q)\mathcal{E}''(\sum_i p_i\varphi_i)\|],\\
 &=\frac{1}{2}\{1-\|\sum_i p_i [q\mathcal{E}'(\varphi_i)-(1-q)\mathcal{E}''(\varphi_i)]\|\},\\
 &\geq \sum_i p_i \{\frac{1}{2}[1-\|q\mathcal{E}'(\varphi_i)-(1-q)\mathcal{E}''(\varphi_i)\|]\},\\
 &=\sum_i  p_i P_{\rm error}^i.
\end{aligned}
\end{equation}
Still one can always find the specific pure state in the decomposition, whose $P_{\rm error}^i\leq P_{\rm error}$.
As a consequence, we can study the minimums of the two important probabilities just investigating the pure state in the Hilbert space.
\end{proof}
Quantum correlations \cite{modi2012classical} like entanglement, discord are essential resources for quantum communication and computation \cite{nielsen2010quantum}, and also for quantum metrology \cite{BRAUNSTEIN1996135,giovannetti2006quantum}. The efficiency of many tasks can be enhanced utilizing quantum correlation resources, e.g. quantum illumination \cite{lloyd2008enhanced} etc.  So here we also want to study the effect of quantum correlation to IFM and investigate whether quantum correlation can enhance the performance of IFM or not.

The new setup (Fig.~\ref{bipartite_input}), employing bipartite photon input state, is like this: The task of the photon $A$ is still to detect the object as in traditional IFM (Fig.~\ref{optical_setup}) and the photon $B$ remains unchanged in the whole process. However, there may be some quantum correlations between photon $A$ and photon $B$. Because quantum channels on photon $A$ are also quantum channels on combined system $A$ and $B$, $Th.~\ref{t1}$ can also be applied to this input case and we just need to consider the pure state in this barpartite scenario. That is to say, the quantum correlation between photon $A$ and $B$ is just entanglement. Therefore this input case is dubbed as the entangled photon input IFM.

\begin{figure}[tb]
\centering
\resizebox{8cm}{!}{\includegraphics[scale=0.8]{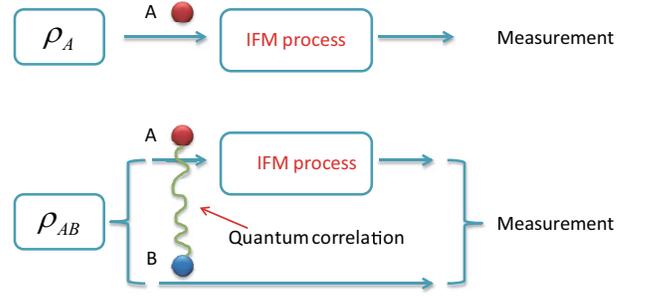}}
\caption{The single photon and entangled photon input IFM}\label{bipartite_input}
\end{figure}

Further more, we will give another theorem below, which describes the relation between the single photon input case and the bipartite photon input case about the two important probabilities $P_{\rm loss}$ and $P_{\rm error}$.
\begin{theorem}\label{t2}
Generally, bipartite photon input case is the same as single photon input case considering $P_{\rm loss}$, but not worse than single photon input case considering $P_{\rm error}$, that is,
\begin{equation}\label{Ploss_bi_single}
P_{\rm loss}(\rho_{AB})=P_{\rm loss}(\rho_{A})
\end{equation}
\begin{equation}\label{Perror_bi_single}
P_{\rm error}(\rho_{AB})\leq P_{\rm error}(\rho_{A})
\end{equation}
where $\rho_{AB}$ is the bipartite input state and $\rho_A=Tr_B(\rho_{AB})$.
\end{theorem}
\begin{proof}
For $P_{\rm loss}$, by the definition in Eq.~\eqref{def_Ploss}, Eq.~\eqref{Ploss_bi_single}
shows
\begin{equation}\label{}
  \langle 3|Tr_B[\mathcal{E}'(\rho_{AB})]|3\rangle=
  \langle 3|\mathcal{E}'(\rho_{A})|3 \rangle,
\end{equation}
which is right since partial trace operation on system $B$ and the quantum channel on system $A$ commute with each other. It means that any bipartite input state $\rho_{AB}$ behaves the same as its corresponding marginal state $\rho_A$ considering $P_{\rm loss}$.

For $P_{\rm error}$, using the definition in Eq.~\eqref{def_Perror}, Eq.~\eqref{Perror_bi_single} is equivalent to
\begin{equation}\label{}
  \|q\mathcal{E}'(\rho_{AB})-(1-q)\mathcal{E}''(\rho_{AB})\|\geq\|q\mathcal{E}'(\rho_{A})-(1-q)\mathcal{E}''(\rho_{A})\|.
\end{equation}
It is also right because partial trace operation on $B$ is certainly a trace-preserving operation that is contractive under the measure of trace distance (see \cite{nielsen2010quantum} and Appendix.~\ref{contractive}), i.e.,
\begin{equation*}\label{}
\begin{aligned}
  &\|q\mathcal{E}'(\rho_{AB})-(1-q)\mathcal{E}(\rho_{AB})\|\\
  &\geq\|q Tr_B[\mathcal{E}'(\rho_{AB})]-(1-q)Tr_B[\mathcal{E}(\rho_{AB})]\|\\
 &=\|q\mathcal{E}'(\rho_{A})-(1-q)\mathcal{E}(\rho_{A})\|
\end{aligned}
\end{equation*}
where the last line is due to partial trace on system B commuting with the quantum channel on system A.
\end{proof}
\section{simplify the quantum channel for pure state}\label{sec3}
Owing to $Th.~\ref{t1}$, we just need to focus on the pure input state scenario and the quantum channels defined in $Sec.~\ref{sec2}$ can be simplified when the input state is pure.

First, the input photon state is set as the general form $|\varphi\rangle=\alpha|1\rangle + \beta|2\rangle$.
Then, for the presence of the object scenario, since the IFM process halts in the probability subspace where the photon decays to the loss state $\ket{3}$, we just need to monitor the probability subspace where the photon is not absorbed; and the corresponding unnormalized state (due to absorption process) is denoted as $\ket{\varphi'}$.
Therefore the final photon state with the object existing is the combination of $|\varphi'\rangle\langle \varphi'|$ (not absorbed part) and $(1-\langle \varphi'|\varphi'\rangle)|3\rangle \langle 3|$ (absorbed part), i.e.,
\begin{equation}
\begin{aligned}
\rho'=\ket{\varphi'}\bra{\varphi'}+(1-\langle \varphi'|\varphi'\rangle)|3\rangle \langle 3|\ .
\end{aligned}
\end{equation}
And for the absence of the object scenario, the final output photon state is denoted by $|\varphi''\rangle$.

Then the quantum channels can be replaced by the corresponding transforming matrices for pure state in the $|1\rangle, |2\rangle$ basis as
\begin{equation}\label{trans_matrix_object}
\left[\left(
  \begin{array}{cc}
    1 & 0 \\
    0 & a \\

  \end{array}
\right)
\left(
  \begin{array}{cc}
    \cos\theta & -\sin\theta \\
    \sin\theta &  \cos\theta \\
  \end{array}
\right)\right ]^N
|\varphi\rangle=|\varphi'\rangle,
\end{equation}
\begin{equation}\label{trans_matrix_nonobject}
\left[
\left(
  \begin{array}{cc}
    \cos\theta & -\sin\theta \\
    \sin\theta &  \cos\theta \\
  \end{array}
\right) \right]^N
|\varphi\rangle=
\left(
  \begin{array}{ccc}
    0 & -1\\
    1 & 0 \\
  \end{array}
\right)|\varphi\rangle=|\varphi''\rangle.
\end{equation}

Here, Eqs.~\eqref{trans_matrix_object}, \eqref{trans_matrix_nonobject} give the relations between $|\varphi\rangle=\alpha|1\rangle + \beta|2\rangle$ and $|\varphi'\rangle$, $|\varphi''\rangle$. Eq.~(\ref{trans_matrix_nonobject}) is just the unitary transformation generated by the rotation operation $R_\theta$ (Eq.~\eqref{rota_3state}) on the photon state in each cycle and we obtain the final state by iterating it for $N$ times. $Eq.~\eqref{trans_matrix_object}$ is not a unitary, describing the decaying of photon to the loss state due to absorption. The state first undergoes the rotation operation, and then the matrix accounting for decay generated by the Kraus operator $A_0$ in Eq.~\eqref{channel_123} operates on it; and the final photon state is also obtained by iterating this process. Actually, $\langle \varphi'|\varphi'\rangle$ is just the conditional probability that the photon is not absorbed conditioning on the object present.

Then $P_{\rm loss}$ and $P_{\rm error}$ can be written in the new form for pure input state as
\begin{equation}\label{Ploss_pure}
\begin{aligned}
 P_{\rm loss}&=q  \ \langle 3|\rho'|3 \rangle,\\
&=q  \ \langle 3||\varphi'\rangle\langle \varphi'|+(1-\langle \varphi'|\varphi'\rangle)|3\rangle \langle 3||3 \rangle,\\
&=q\ (1-\langle \varphi'|\varphi'\rangle),
\end{aligned}
\end{equation}
where the last equality is due to the state $\ket{\varphi'}$ living on the subspace expanded by $|1\rangle, |2\rangle$ basis and having no component on  $|3\rangle$.
\begin{equation}\label{Perror_pure}
\begin{aligned}
&P_{\rm error}=\frac{1}{2}\{1-\|q\rho'- (1-q)|\varphi''\rangle\langle \varphi''| \|\}, \\
&=\frac{1}{2}\{1-\|q[|\varphi'\rangle\langle \varphi'| +
(1-\langle \varphi'|\varphi'\rangle)|3\rangle \langle 3|]- (1-q)|\varphi''\rangle\langle \varphi''| \|\}, \\
&=\frac{1}{2}\{1-\|q |\varphi'\rangle\langle \varphi'| +
P_{\rm loss}|3\rangle \langle 3|]- (1-q)|\varphi''\rangle\langle \varphi''| \|\}, \\
&=\frac{1}{2}\{1-P_{\rm loss}-\|q |\varphi'\rangle\langle \varphi'|
- (1-q)|\varphi''\rangle\langle \varphi''| \|\}.\\
\end{aligned}
\end{equation}
Here the definition of $P_{\rm loss}$ in Eq.~\eqref{Ploss_pure} is employed in the third line; and in the last line, $P_{\rm loss}|3\rangle \langle 3|$ term is extracted from the trace norm since $\ket{\varphi'}, \ket{\varphi''}$ are both in the space spanned by $|1\rangle, |2\rangle$. To be specific, it describes the fact that we will make no error and can always confirm there is an object if the photon is absorbed by it.

The general form of the entangled photon input state is $|\varphi_e\rangle=\alpha|1\rangle|\phi_1\rangle+\beta|2\rangle|\phi_2\rangle$, where $|\phi_1\rangle$, $|\phi_2\rangle$ are any pure states of the photon $B$ part. It is not hard to find that the transforming matrices Eqs.~\eqref{trans_matrix_object}, \eqref{trans_matrix_nonobject} and the $P_{\rm loss}, P_{\rm error}$ expressions Eqs.~\eqref{Ploss_pure}, \eqref{Perror_pure} are also suitable for the entangled photon input case because the transfer matrices play the same role as quantum channels for the photon $A$ part.
In the rest of our article, the single photon and entangled photon input state we consider are $|\varphi_s\rangle=\alpha|1\rangle + \beta|2\rangle$ and $|\varphi_e\rangle=\alpha|1\rangle|\phi_1\rangle+\beta|2\rangle|\phi_2\rangle$ type respectively.

At the end of this section, we show another useful theorem below that describes the condition where the error probability $P_{\rm error}$ can reach $0$. In other words, we can judge whether there is an object without any error.
\begin{theorem}\label{t3}
$P_{\rm error}$ equals to $0$ for pure input state iff $\braket{\varphi''}{\varphi'}=0$.
\end{theorem}

Before we prove Th.~\ref{t3}, let us show a lemma below that is useful to our proof.
\begin{lemma}\label{tr_fidelity_lemma}
Given two pure quantum state, $\ket{\psi_1}$,  $\ket{\psi_2}$ and a positive real number $p$, the following relation hold,
\begin{equation}
  \|p|\psi_1\rangle\langle \psi_1| - |\psi_2\rangle \langle \psi_2| \|=\sqrt{(p+1)^2-4p|\langle \psi_1|\psi_2\rangle|^2}.
 \end{equation}
\end{lemma}
And we leave the proof of Lemma.~\ref{tr_fidelity_lemma} in the Appendix.~\ref{} for conciseness. Now we begin to prove Th.~\ref{t3}.

\begin{proof}
With the definition in Eq.~\eqref{Perror_pure}, $P_{\rm error}$ being
\begin{equation}\label{Perror_zero}
\begin{aligned}
P_{\rm error}=&\frac{1}{2}\{1-P_{\rm loss}-\|q|\varphi'\rangle\langle \varphi'| - (1-q)|\varphi''\rangle \langle \varphi''| \| \},\\
=&\frac{1}{2}\{q\langle\varphi'|\varphi'\rangle+1-q-\|q|\varphi'\rangle\langle \varphi'| - (1-q)|\varphi''\rangle \langle \varphi''| \| \},\\
=&\frac{1}{2}\{q\langle\varphi'|\varphi'\rangle+1-q\\
&-(1-q)\|\frac{q\bra{\varphi'}\varphi'\rangle}{1-q}\frac{|\varphi'\rangle\langle \varphi'|}{\bra{\varphi'}\varphi'\rangle}-|\varphi''\rangle \langle \varphi''| \| \},\\
=&\frac{1}{2}\{q\langle\varphi'|\varphi'\rangle+1-q\\
&-\sqrt{(q\langle\varphi'|\varphi'\rangle+1-q)^2-4q(1-q)|\langle\varphi''|\varphi'\rangle|^2}\},\\
\end{aligned}
\end{equation}
where in the second line we apply the definition of $P_{\rm loss}$ in Eq.~\eqref{Ploss_pure}; and in the last line we employ Lemma.~\ref{tr_fidelity_lemma}, by substituting $\frac{\ket{\varphi'}}{\sqrt{\bra{\varphi'}\varphi'\rangle}}$, $\ket{\varphi''}$ for $\ket{\psi_1}$, $\ket{\psi_2}$, and $\frac{q\bra{\varphi'}\varphi'\rangle}{1-q}$ for $p$.
Then, let us observe the last line in Eq.~\eqref{Perror_zero}:
the second part in the square root, i.e., $4q(1-q)|\langle\varphi''|\varphi'\rangle|^2$, is non-negative, so it is not hard to find that $P_{\rm error}$ can reach $0$ iff $\braket{\varphi''}{\varphi'}=0$.
\end{proof}
\section{opaque object case}\label{sec4}
We study IFM of opaque object with finite interrogation cycle $N$ in this section.
Our task is to use the model simplified in $Sec.~\ref{sec3}$ to find the minimal values of the two important probabilities $P_{\rm loss}$ and $P_{\rm error}$, and the corresponding states to reach them. When the object is opaque, i.e., $a=0$, $Eq.~\eqref{trans_matrix_object}$ shows:
\begin{equation}\label{matrix_absorb_opaque}
\left(
  \begin{array}{cc}
    \cos^N\theta & -\sin\theta\cos^{N-1}\theta \\
    0 &  0 \\

  \end{array}
\right)
|\varphi\rangle=|\varphi'\rangle.
\end{equation}
\subsection{$P_{\rm loss}$ and $P_{\rm error}$ study with single photon input state }\label{subs41}
First, let us focus on the loss probability $P_{\rm loss}$.
Setting the input state as $|\varphi\rangle=\alpha|1\rangle + \beta|2\rangle$ and using $Eqs.~\eqref{matrix_absorb_opaque}, \eqref{Ploss_pure}$, we can get
\begin{equation}\label{state_object}
|\varphi'\rangle=\cos^{N-1} \theta(\alpha \cos \theta-\beta \sin \theta)|1\rangle,
 \end{equation}
\begin{equation}\label{Ploss_opaque_inequality}
\begin{aligned}
P_{\rm loss}&=q(1-|\cos^{N-1} \theta(\alpha \cos \theta-\beta \sin \theta)|^2), \\
&\geq q(1-\cos^{2(N-1)} \theta),
\end{aligned}
 \end{equation}
where the inequality in the second line of Eq.~\eqref{Ploss_opaque_inequality} is due to the fact that the absolute value of the inner product for the two vectors $(\cos \theta, -\sin \theta)^T$ and $(\alpha, \beta)^T$ is not larger than $1$.
And it is not hard to find that the minimum can be reached by the state $|\varphi_a\rangle= \cos \theta |1\rangle - \sin \theta|2\rangle$, with the global phase neglected (since the global phase actually makes no difference to the state in the IFM process, we always neglect it without announcement in the following).

Then we study the error probability $P_{\rm error}$. And we calculate the value of $\langle \varphi'' |\varphi'\rangle$ and check whether there are input states can fulfill the condition stated in $Th.~\ref{t3}$ and let $P_{\rm error}$ reach $0$. With the help of $Eq.~\eqref{trans_matrix_nonobject}$ and $Eq.~\eqref{state_object}$, we have
\begin{equation}\label{}
  |\varphi''\rangle=-\beta |1\rangle+\alpha |2\rangle,
 \end{equation}
\begin{equation}\label{}
 \langle \varphi'' |\varphi'\rangle=-\beta^* \cos^{N-1} \theta(\alpha \cos \theta-\beta \sin \theta).
 \end{equation}

It is not difficult to find that there are two states can make $\langle \varphi'' |\varphi'\rangle=0$. The first one is $|\varphi_b\rangle=|1\rangle$, and the second one is $|\varphi_c\rangle=\sin \theta |1\rangle + \cos \theta |2\rangle$. That is, we can both realize zero error in IFM with these two states.

Thus, it is necessary to compare the loss probability $P_{\rm loss}$ of $|\varphi_b\rangle$, $|\varphi_c\rangle$. And the $P_{\rm loss}$ of the two states are $ q(1-\cos^{2N} \theta)$ and $q$ respectively, by the definition of $Eq.~\eqref{Ploss_pure}$. It indicates that the first state is better than the second one considering $P_{\rm loss}$, since when $N$ is large enough,
\begin{equation}\label{}
 q(1-\cos^{2N} \theta) \simeq q\frac{\pi^2}{4N} \ll q.
 \end{equation}
And note that the $P_{\rm loss}$ of $|\varphi_b\rangle$ will approaches $0$, as $N\rightarrow \infty$~~\cite{kwiat1995interaction,kwiat1999high}.  For the second state $|\varphi_c\rangle$, the photon is always lost if the object is there; it is the reason why $|\varphi_c\rangle$ can detect the object without any error. But it is useless in our problem since it violates the principle of IFM, i.e., detecting the object with as small as possible photon loss probability.
\subsection{$P_{\rm loss}$ and $P_{\rm error}$ study with entangled photon input state}\label{subs42}
Here we want to study the power of quantum correlation for IFM in the opaque object case. So we set the initial input state as the general form $|\varphi\rangle=\alpha|1\rangle|\phi_1\rangle+\beta|2\rangle|\phi_2\rangle$ , where $|\phi_1\rangle$, $|\phi_2\rangle$ are any pure states of the photon $B$ part and $\alpha$, $\beta$ are non-negative real numbers (one can always remove the phase information in $\alpha$, $\beta$ to the states $|\phi_1\rangle$, $|\phi_2\rangle$ of the photon $B$ part to obtain this form).

As mentioned earlier, the equations utilized in the single photon input case can also be used in this entangled photon input case. And we should do the transforming matrix operations on the photon $A$ part and evaluate the two probabilities $P_{\rm loss}$ and $P_{\rm error}$ in the same way as in $~\ref{subs41}$.

We first study the loss probability $P_{\rm loss}$. Using $Eq.~\eqref{matrix_absorb_opaque}$, we have
\begin{equation}\label{state_object_entangle}
 |\varphi'\rangle=\cos ^{N-1}\theta|1\rangle(\alpha \cos \theta|\phi_1\rangle-\beta \sin \theta|\phi_2\rangle),
 \end{equation}
and with the help of $Eq.~\eqref{Ploss_pure}$, the loss probability shows
\begin{equation}\label{}
\begin{aligned}
  P_{\rm loss}&=q[1-|\cos ^{N-1}\theta(\alpha \cos \theta|\phi_1\rangle-\beta \sin \theta|\phi_2\rangle)|^2],\\
  &\geq q[1-\cos^{2(N-1)} \theta(\alpha \cos \theta+\beta \sin \theta)^2],\\
  &\geq q[1-\cos^{2(N-1)} \theta].
\end{aligned}
\end{equation}
Here the first inequality is saturated when $|\phi_1\rangle=-|\phi_2\rangle$,  and the second inequality is saturated when $ \alpha=\cos \theta$ and $\beta=\sin \theta$ . Thus, the minimum of $P_{\rm loss}$ can be reached by the state $|\varphi_a^*\rangle=(\cos\theta|1\rangle-\sin\theta|2\rangle)|\phi_1\rangle$, where we use the superscript $*$ label the bipartite state.
Especially, it is a product state, which is equivalent to the state $|\varphi_a\rangle$  in the single photon input case after neglecting the photon $B$ part.

Next, we study the error probability $P_{\rm error}$  in this entangled photon input case. With $Eq.~\eqref{state_object}$, we can obtain the output state in the absence of the object as
\begin{equation}\label{}
  |\varphi''\rangle=\alpha |2\rangle|\phi_1\rangle-\beta |1\rangle|\phi_2\rangle,
 \end{equation}
and applying Eq.~\eqref{state_object_entangle}, the value of $\langle \varphi'' |\varphi'\rangle$ shows the following form
\begin{equation}\label{innerp_opaque_entangle}
  \langle \varphi'' |\varphi'\rangle
  =(-\alpha\beta\cos\theta\langle\phi_2|\phi_1\rangle+\beta^2\sin\theta)\cos^{N-1} \theta.
\end{equation}
From Eq.~\eqref{innerp_opaque_entangle}, we can get a family of the solutions for $\langle \varphi'' |\varphi'\rangle=0$ which satisfies
\begin{equation}\label{family_perror_0}
  \frac{\beta\sin\theta}{\alpha\cos\theta}=\langle\phi_2|\phi_1\rangle.
 \end{equation}
The two solutions in the single photon input case are both included in  $Eq.~\eqref{family_perror_0}$. They are the states $|\varphi_b^*\rangle=|1\rangle|\phi_1\rangle$ and $|\varphi_c^*\rangle=(\sin \theta |1\rangle + \cos \theta |2\rangle)|\phi_1\rangle$. Now we shall check which one is the best state in this family considering $P_{\rm loss}$. Using $Eqs.~\eqref{Ploss_pure},~\eqref{family_perror_0}$, we get
\begin{equation}\label{}
\begin{aligned}
  P_{\rm loss}
  =q[1-\cos \theta^{2(N-1)}(\alpha^2 \cos^2\theta-\beta^2 \sin ^2\theta)].
\end{aligned}
\end{equation}
It is clear that $|\varphi_b^*\rangle=|1\rangle|\phi_1\rangle$ reaches the minimum $q(1-\cos^{2N} \theta)$ in this family, which is equivalent to $|\varphi_b\rangle$ in the single photon input case.

From $~\ref{subs41},~\ref{subs42}$, we conclude that the entangled photon input state makes no enhancement to the optimization for the two important probabilities $P_{\rm loss}$ and $P_{\rm error}$ respectively, compared with single photon input state, and the states which reach the minimums are the same in some sense in these two cases. In addition, the state $|1\rangle$ is the optimal state which can make $P_{\rm loss}$ and $P_{\rm error}$ both reach zero when $N\rightarrow \infty$.
\section{semitransparent object case}\label{sec5}
In this section, we go further for the general case. In practical application of IFM, the object is always semitransparent, i.e., partially absorbing the photon. Thus, we will study the minimal $P_{\rm loss}$ and $P_{\rm error}$, and the states to reach them also in this semitranparent object case, just like in the opaque object case. In addition, the effect of quantum entanglement is also investigated.
\subsection{simplify the transforming matrix}\label{subs51}
The major difficulty to study the general case is to simplify the transforming matrix in $Eq.~\eqref{trans_matrix_object}$. First we can represent the matrix in one interrogation cycle with Pauli matrices as
\begin{equation}\label{}
\begin{aligned}\
  C_0=&\left(
  \begin{array}{cc}
    1 & 0 \\
    0 & a \\
  \end{array}
\right)
\left(
  \begin{array}{cc}
    \cos\theta & -\sin\theta \\
    \sin\theta &  \cos\theta \\
  \end{array}
\right),\\
  =&\frac{(1-a)\cos\theta}{2} \sigma_z
-\frac{i(1+a)\sin\theta}{2} \sigma_y - \frac{(1-a)\sin\theta}{2} \sigma_x \\
&+\frac{(1+a)\cos\theta}{2}I. \\
\end{aligned}
\end{equation}

Then we changes the basis by applying a unitary transformation $U=e^{-i\frac{\sigma_y}{2}\theta}$ and obtain
\begin{equation}\label{newbasis_objexist}
\begin{aligned}\
C_1=&UC_0U^\dag,\\
=&\frac{(1-a)}{2} \sigma_z
-\frac{i(1+a)\sin\theta}{2} \sigma_y+\frac{(1+a)\cos\theta}{2}I,\\
=&\frac{(1-a)}{2}(\sigma_z-ik_1\sigma_y+k_2I),\\
\end{aligned}
\end{equation}
where we use
\begin{equation}\label{k1_k2}
\begin{aligned}\
k_1=\frac{(1+a)\sin\theta}{1-a},\\
k_2=\frac{(1+a)\cos \theta}{1-a},
\end{aligned}
\end{equation}
for simplicity, and they are both positive numbers.
Because we sample all the states in the Hilbert space, the change of the basis or the unitary transformation does not matter. Thus hereafter, we handle the semitransparent object scenario IFM in the new basis for simplicity. And in the following, any matrix $O$ should be changed to $UOU^\dag$; $\ket{\varphi}$ labels a specific vector coordinate in the new basis for convenience and $U^\dag\ket{\varphi}$ is the same state but the coordinate value is obtained in the old basis.

The power $N$ of the matrix $C_1$, labeled by $C$, can be calculated by expanding the binomial with the help of the equality $(\sigma_z-ik_1\sigma_y)^2=1-k_1^2$, which is the result of the anti-commutation relation $\{\sigma_z,\sigma_y\}=0$.
\begin{equation}\label{C_newbasis}
\begin{aligned}
C=&C_1^N,\\
 =&(\frac{1-a}{2})^N [(\sigma_z-ik_1\sigma_y)+k_2I)]^N,\\
 =&(\frac{1-a}{2})^N [\sum_{k\in odd}
\binom{N}{k}(1-k_1^2)^{\frac{k-1}{2}}k_2^{N-k}(\sigma_z-ik_1\sigma_y)\\
&+\sum_{k\in even}
\binom{N}{k}(1-k_1^2)^{\frac{k}{2}}k_2^{N-k}I],\\
  =&(\frac{1-a}{2})^N [f_1(\sigma_z-ik_1\sigma_y)+f_2I],
\end{aligned}
\end{equation}
where we substitute $f_1$ and $f_2$ for the summations before the operators $(\sigma_z-ik_1\sigma_y)$ and $I$ respectively. In fact, $f_1$ and $f_2$ are related to the summations of the even and odd terms in the corresponding binomial.

Thus, we define $\Sigma_1$ and $\Sigma_2$ as below, which are sum of the odd and even terms of the corresponding binomial. When $k_1\leq1$:
\begin{equation}\label{sum_k1small}
  \left\{
    \begin{aligned}
      &\Sigma_1=\frac{(\sqrt{1-k_1^2}+k_2)^N-(-\sqrt{1-k_1^2}+k_2)^N}{2}\\
      &\Sigma_2=\frac{(\sqrt{1-k_1^2}+k_2)^N+(-\sqrt{1-k_1^2}+k_2)^N}{2}
    \end{aligned}
  \right.
\end{equation}
when $k_1>1$:
\begin{equation}\label{sum_k1large}
  \left\{
    \begin{aligned}
      &\Sigma_1=\frac{(i\sqrt{k_1^2-1}+k_2)^N-(-i\sqrt{k_1^2-1}+k_2)^N}{2}\\
      &\Sigma_2=\frac{(i\sqrt{k_1^2-1}+k_2)^N+(-i\sqrt{k_1^2-1}+k_2)^N}{2}
    \end{aligned}
  \right.
\end{equation}
Then we can obtain the expressions for $f_1$ and $f_2$ in $Eq.~\eqref{C_newbasis}$ with $\Sigma_1$ and $\Sigma_2$, when $k_1\leq1$:
\begin{equation}\label{f_k1small}
  \left\{
    \begin{aligned}
      &f_1=\frac{\Sigma_1}{\sqrt{1-k_1^2}}\\
      &f_2=\Sigma_2
    \end{aligned}
  \right.
\end{equation}
when $k_1>1$:
\begin{equation}\label{f_k1large}
  \left\{
    \begin{aligned}
      &f_1=\frac{\Sigma_1}{i\sqrt{k_1^2-1}}\\
      &f_2=\Sigma_2
    \end{aligned}
  \right.
\end{equation}

The insight of the above result is that the eigenstates of $C_1$ and $C$ should be the same and the eigenvalues from $C$ are just the power $N$ of the ones from $C_1$ . So the structures of $Eq.~\eqref{newbasis_objexist},~\eqref{C_newbasis}$  are also the same, linear combination of $(\sigma_z-ik_1\sigma_y)$ and $I$. Especially, $(\sigma_z-ik_1\sigma_y)$ determines the eigenstates and the eigenvalues of it are $\pm\sqrt{1-k_1^2}$. That's why we have the formulas like $Eq.~\eqref{sum_k1small},~\eqref{sum_k1large},~\eqref{f_k1small},~\eqref{f_k1large}$ . Clearly, $f_1$ and $f_2$ are functions of $a$ and $\theta$ and we will show that they are both real positive number in the following theorem.

\begin{theorem}\label{}
 $f_1$ and $f_2$ are both real positive numbers no matter what value $k_1$ is.
\end{theorem}

\begin{proof}
When $k_1\leq1$, $\Sigma_1$ and $\Sigma_2$ are the sum of odd and even terms of $(\sqrt{1-k_1^2}+k_2)^N$ respectively. It is obvious that $f_1$ and $f_2$ are both real positive numbers. When $k_1>1$, $\Sigma_1$ and $\Sigma_2$ are the imaginary and real part of $(i\sqrt{k_1^2-1}+k_2)^N$. We just need to check which quadrant this complex number locates in. Because $\frac{\sqrt{k_1^2-1}}{k_2}\leq \frac{k1}{k2}=\tan\theta$ and $N \theta=\frac{\pi}{2}$, we know it locates in the first quadrant. Then  $f_1$ and $f_2$ are also real positive numbers in this case by the definition Eq.~\eqref{f_k1large}.
\end{proof}
\subsection{$P_{\rm loss}$ study with single photon and entangled photon input state}\label{subs52}
With the knowledge of $~\ref{subs51}$, now we can get the photon loss probability $P_{\rm loss}$ in the new basis by the definition Eq.~\eqref{Ploss_pure} as
\begin{equation}\label{Ploss_newbasis}
\begin{aligned}
P_{\rm loss}&=q(1- \langle\varphi'|\varphi'\rangle),\\
&=q(1- \langle\varphi| C^\dag C|\varphi\rangle),\\
&=q[1- \mathrm{Tr}_{AB}(C^\dag C|\varphi\rangle\langle\varphi|)],\\
&=q[1- \mathrm{Tr}_{A}(C^\dag C\rho_A)],
\end{aligned}
\end{equation}
where in the final line we trace out the photon $B$ part since the transforming matrix $C$ just operates on the photon $A$.
Eq.~\eqref{Ploss_newbasis} reminds us that the entangled photon input state $|\varphi_{AB}\rangle$ behaves the same as $\mathrm{Tr}_B(\varphi_{AB})=\rho_A$ for $P_{\rm loss}$, as showed in $Th.~\ref{t2}$. Especially, if one reaches the minimum of $P_{\rm loss}$ with the single photon input state $|\varphi_A\rangle$, one can surely find any pure state like $|\varphi_A\rangle|\phi_B\rangle$ to reach the same minimal value. Hence we just need to study $P_{\rm loss}$ in the single photon input case.

Thus $\langle\varphi|C^\dag C|\varphi\rangle$ in $Eq.~\eqref{Ploss_newbasis}$ should be maximized only for single photon input state, and $C^\dag C$ can be expanded as
\begin{equation}\label{CdagC}
\begin{aligned}
C^\dag C=
(\frac{1-a}{2})^{2N}[f_1^2(1+k_1^2)+f_2^2]I+2f_1(f_2\sigma_z-f_1k_1\sigma_x).
\end{aligned}
\end{equation}
It is the same as to find the larger eigenvalue for a single spin Hamiltonian. Thus, no matter what the value of $k_1$ is, it is not hard to obtain the minimal $P_{\rm loss}$ being
\begin{equation}\label{Plossmin_newbasis_semi}
\begin{aligned}
(P_{\rm loss})_{min}=q[1-(\frac{1-a}{2})^{2N}(f_1+\sqrt{f_2^2+f_1^2k_1^2})^2].
\end{aligned}
\end{equation}

Utilizing $Eq.~\eqref{Plossmin_newbasis_semi}$, the relation between the normalized photon loss rate $(P_{\rm loss}/q)_{min}$ and the interrogation cycle $N$ for different transparency $a^2$ is exhibited in $Fig.~\ref{Ploss_N_a}$. It shows that when $N$ is large enough, $(P_{\rm loss}/q)_{min}$ decreases with the increasing of $N$ no matter what value $a$ is. Generally speaking, $(P_{\rm loss}/q)_{min}$ of small $a$ is always less than that of large $a$ for a fixed large enough $N$. However, $(P_{\rm loss}/q)_{min}$ can increase and then decrease for large enough $a$ with the increasing of $N$. Via numerical analysis, we find that the maximum of the curve for a given large $a$ can be obtained at $N'$, which is slightly larger than the one determined by the equation $k_1=\frac{1+a}{1-a}\sin(\frac{\pi}{2N})=1$, as showed in
Fig.~\ref{Perror_0_regin}.
\begin{figure}[t!]
\centering
\resizebox{7.5cm}{!}{\includegraphics[scale=0.8]{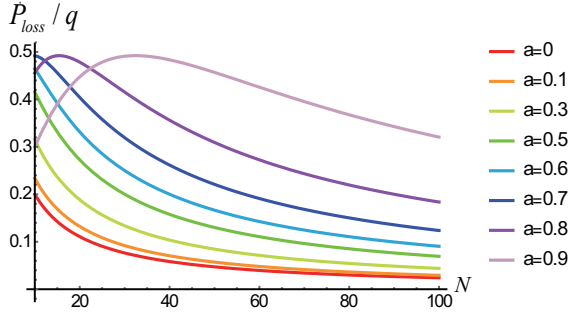}}
\caption{$(P_{\rm loss}/q)_{min}$ vs the interrogation cycle $N$ for different transparency $a^2$}\label{Ploss_N_a}
\end{figure}

The state reaching the minimum of $P_{\rm loss}$, named $|\varphi_0\rangle$, is just the eigenstate of $C^\dag C$ with larger eigenvalue. $|\varphi_0\rangle\langle\varphi_0|$ is on the $xz$ plane of the Bloch sphere with the angle between state the $|\varphi_0\rangle\langle\varphi_0|$ and the $z$ direction is $\theta_1=\arctan(\frac{f_1k_1}{f_2})$ (see Fig.~\ref{Blochs_position}). The corresponding vector
$U^\dag |\varphi_0\rangle$ is the one which reach the minimal $P_{\rm loss}$ in the old basis. And it is not hard to find $U^\dag |\varphi_0\rangle=\ket{\varphi_a}$ when the transparency $a^2=0$, i.e., opaque object case.
\subsection{$P_{\rm error}$ study with single photon and entangled photon input state}\label{subs53}
Here, we derive the error probability $P_{\rm error}$ of IFM in both single photon and entangled photon input scenarios.

For convenience, we label the unitary transformation in $Eq.~\eqref{trans_matrix_nonobject}$ by $D$, as the object is absent.
\begin{equation}\label{}
D=\left(
  \begin{array}{cc}
    0 & -1 \\
    1 &  0 \\
  \end{array}
\right)=-i\sigma_y.
\end{equation}

In $Th.~\ref{t3}$, we have showed that $P_{\rm error}$ can reach $0$ iff $\langle\varphi''|\varphi'\rangle=0$, no matter which type the input state is. By definition, we get $\langle\varphi''|\varphi'\rangle$ in the new basis as
\begin{equation}\label{innerp_newbasis}
\begin{aligned}
\langle\varphi''|\varphi'\rangle=&\langle\varphi|U D^\dag U^\dag C|\varphi\rangle,\\
=&\langle\varphi|D^\dag C|\varphi\rangle,\\
=&\mathrm{Tr}_{AB}(D^\dag C |\varphi\rangle\langle\varphi|),\\
=&\mathrm{Tr}_A(D^\dag C \rho_A).
\end{aligned}
\end{equation}
In the second line, we use the fact that $U$ commutes with $D^\dag$; The third line is due to the fact that $D^\dag$ and $C$ only operate on the photon $A$ part. From the definition of  $D$ and $C$, we have $D^\dag C$ being
\begin{equation}\label{DdagC}
\begin{aligned}
D^\dag C
=\left(\frac{1-a}{2}\right)^{N}[f_1k_1I+(i f_2\sigma_y -f_1\sigma_x)].
\end{aligned}
\end{equation}
In the meantime,  $\rho_A$ has the following Bloch sphere representation,
\begin{equation}\label{two_level_state}
\rho_A=\frac{1}{2}(I+\vec{r}\cdot \vec{\sigma}).
\end{equation}
Then Eq.~\eqref{innerp_newbasis} becomes $\langle\varphi''|\varphi'\rangle=(\frac{1-a}{2})^{N}(f_1k_1-f_1r_x+if_2r_y)$ with the fact that the trace of pauli matrix is $0$. In order to make $P_{\rm error} = 0$, we should let $r_x=k_1$ and $r_y=0$.

\begin{figure}[t!]
\centering
\resizebox{6cm}{!}{\includegraphics[scale=0.6]{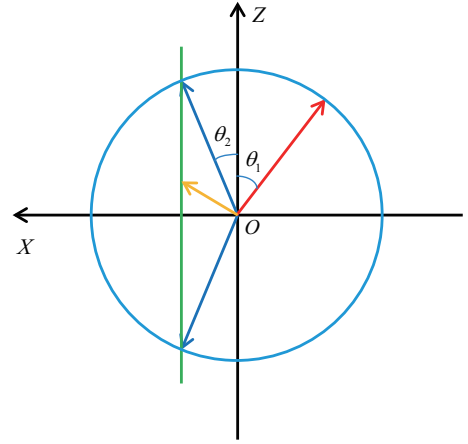}}
\caption{The positions of the states on Bloch sphere which reach the minimal $P_{\rm loss}$ and $P_{\rm error}$ in the new basis. The red vector represents the state $|\varphi_0\rangle\langle\varphi_0|$ which reaches the minimal $P_{\rm loss}$. The two blue vectors represent $|\varphi_\pm\rangle\langle\varphi_\pm|$. Any mixed state on the green line which is the connection between the end of the two blue vectors can satisfy  $\mathrm{Tr}_A(D^\dag C \rho_A)=0$. The yellow vector represents one of these mixed states and its purification is a entangled photon input state $\alpha\ket{\varphi_+}|\phi_1\rangle+\beta\ket{\varphi_-}|\phi_1^\perp\rangle$
, which makes $P_{\rm error}=0$.}\label{Blochs_position}
\end{figure}

When $k_1\leq 1$, there are two pure state solutions $\rho_A=|\varphi_\pm\rangle\langle\varphi_\pm|=\frac{1}{2}(I+k_1\sigma_x\pm\sqrt{1-k_1^2}\sigma_z)$ of photon A. The angle between each pure solution $|\varphi_\pm\rangle\langle\varphi_\pm|$ and the $z$ axis is $\theta_2=\arctan(\frac{k_1}{\sqrt{1-k_1^2}})$  on the Bloch sphere (see $Fig.~\ref{Blochs_position}$). And it is straightforward to see that any convex mixing of the two pure solutions can also lead to $Tr_A(D^\dag C \rho_A)=0$. Therefore, in the bipartite scenario, the solution to $P_{\rm error}=0$ is $\alpha\ket{\varphi_+}|\phi_1\rangle+\beta\ket{\varphi_-}|\phi_1^\perp\rangle$, where $\bra{\phi_1^\perp}\phi_1\rangle = 0$ and $\alpha, \beta$ are two arbitrary state coefficients.
Like in $a=0$ case, we have a family of best states which reach $P_{\rm error}=0$ in the entangled photon input scenario.

Furthermore, we aim to find the solution that minimize the photon loss rate $P_{\rm loss}$ given in $Eq.~\eqref{Ploss_newbasis}$ in this family. Combining the solution to $P_{\rm error}=0$, we can show that the optimal state in this family is $|\varphi_+\rangle\langle\varphi_+|=\frac{1}{2}(I+k_1\sigma_x+\sqrt{1-k_1^2}\sigma_z)$ with the minimal $P_{\rm loss}$ value being
\begin{equation}\label{Ploss_familymin}
\begin{aligned}
(P_{\rm loss})_{|\varphi_+\rangle}=q[1-(\frac{1-a}{2})^{2N}(f_1\sqrt{1-k_1^2}+f_2)^2].
\end{aligned}
\end{equation}
which means entangled photon input state does no good to $P_{\rm error}$ in this $k_1$ regime.
\begin{figure}[t!]
\centering
\resizebox{6.5cm}{!}{\includegraphics[scale=0.8]{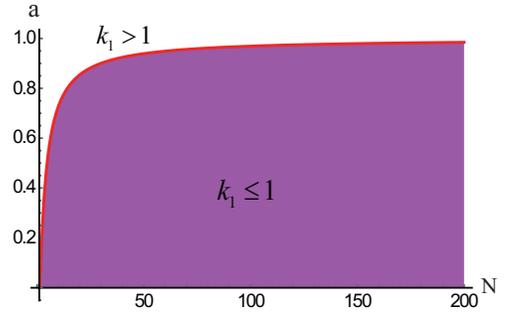}}
\caption{Red curve: transparency $a$ vs interrogation cycle $N$ determined by $k_1=\frac{1+a}{1-a}\sin(\frac{\pi}{2N})=1$. Shadow purple region indicates the parameter domain where we can reach $P_{\rm error}=0$.}\label{Perror_0_regin}
\end{figure}

When $k_1>1$, there is no solution to $\langle\varphi''|\varphi'\rangle=0$ or equivalently $P_{\rm error}=0$. Nevertheless, we can still analyze the nonzero minimum of $P_{\rm error}$. Using Eqs.~\eqref{Perror_zero},~\eqref{Ploss_newbasis} and~\eqref{innerp_newbasis}, we have the general expression of $P_{\rm error}$ being
\begin{equation}\label{perror_general}
\begin{aligned}
P_{\rm error}&=\frac{1}{2}\{qTr[C^\dag C\rho_A]+(1-q)\\
&-\sqrt{(qTr[C^\dag C\rho_A]+1-q)^2-4q(1-q)|Tr[D^\dag C\rho_A]|^2}\},
\end{aligned}
\end{equation}
which is suitable no matter what value $k_1$ is.
It indicates that $|\varphi_{AB}\rangle$ appears in the form $Tr_B(\varphi_{AB})=\rho_A$ for $P_{\rm error}$ in all $k_1$ regime. It is crucial to emphasize that the expression for $P_{\rm error}$ of Eq.~\eqref{perror_general} is suitable for any pure states, single photon or entangled photon input, but not for mixed state $\rho_A$, because our pure state prerequisite.
Moreover, we find the entangled photon input state can not enhance the performance on $P_{\rm error}$ for any values of $k_1$, compared with the single photon input state, i.e., the minimum of Eq.~\eqref{perror_general} should be reached by pure state $\rho_A=\ket{\varphi_A}\bra{\varphi_A}$. The detailed discussion about the effect of the quantum correlation to $P_{\rm error}$  is arranged in Appendix.~\ref{entanglement_enhance}.
\subsection{$N\rightarrow\infty$ behavior}\label{subs54}
In the above subsections, we have systematically analysed the general IFM model of the semitransparent object with finite interrogation cycle. Now, in this part, we study the asymptotic behavior of the relevant quantities when the interrogation cycle $N\rightarrow\infty$.
The behavior of the minimal values for $P_{\rm loss}$, $P_{\rm error}$ and the initial input states which can reach the minimums are investigated in the $N\rightarrow\infty$ condition.

When the interrogation cycle $N\rightarrow\infty$, $k_1=\frac{1+a}{1-a}\sin(\frac{\pi}{2N})\rightarrow0<1$ for any fixed $a$. Therefore we always have the state $|\varphi_+\rangle$ to reach $P_{\rm error}=0$. First we consider the asymptotic behavior of $(P_{\rm loss})_{|\varphi_+\rangle}$, described by $Eq.~\eqref{Ploss_familymin}$. With the help of $Eqs.~\eqref{sum_k1small}, ~\eqref{f_k1small}$ and the definitions of $k_1$, $k_2$ (Eq.~\eqref{k1_k2}), we have
\begin{equation}\label{approximate_1}
\begin{aligned}
&(\frac{1-a}{2})^{2N}(f_1\sqrt{1-k_1^2}+f_2)^2,\\
=&(\frac{1-a}{2})^{2N}(\Sigma_1+\Sigma_2)^2,\\
=&[\frac{1-a}{2}(k_2+\sqrt{1-k_1^2})]^{2N},\\
=&[\frac{(1+a)\cos\theta+\sqrt{(1-a)^2-(1+a)^2\sin^2\theta}}{2}]^{2N},\\
\simeq&[1-\frac{1+a}{1-a}\frac{\pi^2}{8N^2}+O(\frac{1}{N^4})]^{2N},\\
\simeq&1-\frac{1+a}{1-a}\frac{\pi^2}{4N}+O(\frac{1}{N^3}),
\end{aligned}
\end{equation}
where we use the fact the $\cos\theta=1-\frac{\theta^2}{2}+O(\theta^4)$, $\sin\theta=\theta-O(\theta^3)$ and $\theta=\frac{\pi}{2N}$. Then the asymptotic expression of $Eq.~\eqref{Ploss_familymin}$ is
\begin{equation}\label{Ploss_familymin_asym}
\begin{aligned}
(P_{\rm loss})_{|\varphi_+\rangle}^{N\rightarrow\infty}\simeq q[\frac{1+a}{1-a}\frac{\pi^2}{4N}-O(\frac{1}{N^3})].
\end{aligned}
\end{equation}
Clearly, whatever the value of $a$ is, $(P_{\rm loss})_{|\varphi_+\rangle}$ goes to $0$ for sufficient large $N$.

Furthermore, we aim to consider the asymptotic behavior of $Eq.~\eqref{Plossmin_newbasis_semi}$, the minimum of $P_{\rm loss}$. Utilizing the similar approximation technique as for $(P_{\rm loss})_{|\varphi_+\rangle}$, it shows
\begin{equation}\label{Plossmin_asym}
\begin{aligned}
(P_{\rm loss})_{min}^{N\rightarrow\infty}\simeq q[\frac{1+a}{1-a}\frac{\pi^2}{4N}-O(\frac{1}{N^2})].
\end{aligned}
\end{equation}
And the detailed derivation is put in the Appendix.~\ref{Plossmin_asym_App}. In addition,
the asymptotic behavior of $(P_{\rm loss})_{|\varphi_+\rangle}/q, (P_{\rm loss})_{min}/q$ have been plotted in $Fig.~\ref{asymptotic}$.
\begin{figure}[t!]
\centering
\resizebox{9cm}{!}{\includegraphics[scale=0.6]{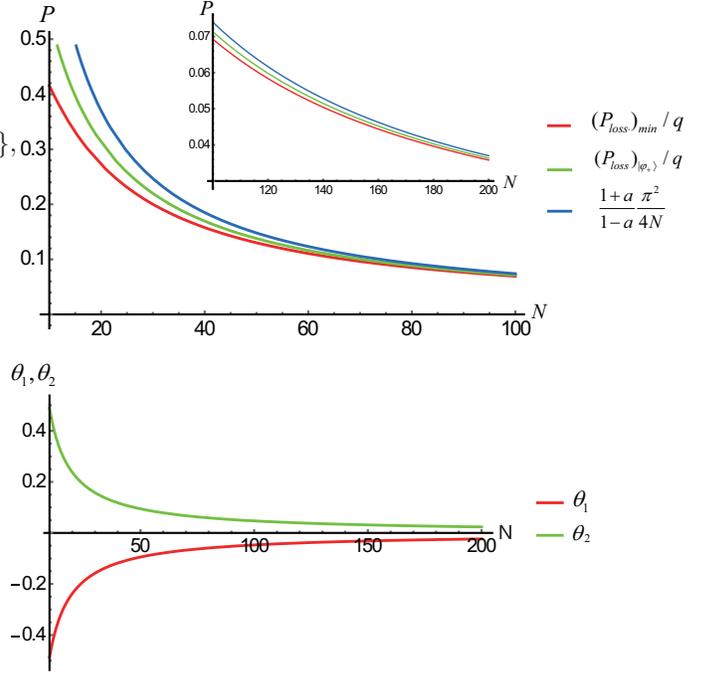}}
\caption{All the graphs are plotted at $a=0.5$. (a) The asymptotic behaviors of $(P_{\rm loss})_{|\varphi_+\rangle}/q, (P_{\rm loss})_{min}/q$ as $N\rightarrow\infty$. The green line $(P_{\rm loss})_{|\varphi_+\rangle}/q$ is always above the red line $(P_{\rm loss})_{min}/q$. The main term of the asymptotic expressions in Eqs.~\eqref{Ploss_familymin_asym},~\eqref{Plossmin_asym}, i.e., $\frac{1+a}{1-a}\frac{\pi^2}{4N}$, is also shown in the plot with blue line. Inset: N ranges from 100 to 200. All the three terms go to zero asymptotically when $N\rightarrow\infty$.
(b) The asymptotic behavior of $\theta_1$ and $\theta_2$. We use negative sign for $\theta_1$ because it locates at the negative X axis side, as showed in Fig.~\ref{Blochs_position}.}\label{asymptotic}
\end{figure}

When $N\rightarrow\infty$, $\theta_1$ and $\theta_2$, relating to the initial input states $|\varphi_0\rangle$, $|\varphi_+\rangle$, both go to zero (see Fig.~\ref{asymptotic}). And the unitary $U=e^{-i\frac{\sigma_y}{2}\theta}$ of changing basis goes to identity. Hence, the corresponding vectors $U^\dag|\varphi_0\rangle$ which reaches the $(P_{\rm loss})_{min}$ and $U^\dag|\varphi_+\rangle$ which reaches the minimum of $P_{\rm loss}$ but keeping $P_{\rm error}=0$ in the old basis , go to the same vector $(1,0)^T$, i.e., $|1\rangle$ in our system. That is to say, as $N\rightarrow\infty$, we can use $|1\rangle$ to realize $P_{\rm loss}=P_{\rm error}=0$ asymptotically, perfect detecting the object without photon loss even if the object is a semitransparent one.

\section{conclusion and outlook}\label{sec6}
In conclusion, with the help of quantum channel theory, we build the general model of quantum-Zeno-like IFM, where the object to be detected is semitransparent and the the number of interrogation cycle is finite. Two important probabilities named $P_{\rm loss}$ and $P_{\rm error}$ are proposed to describe the photon loss rate and the error of discrimination in the IFM process. In order to find the minimums of the $P_{\rm loss}$ and $P_{\rm error}$ and the corresponding initial photon input states to reach them, we simplify the iteration of the quantum channels to transfer matrices operating on pure state. With this compact simplification, the minimum properties of $P_{\rm loss}$ and $P_{\rm error}$ can be systemically studied. In addition, it shows that the entangled photon input state can not enhance the performance of IFM, considering $P_{\rm loss}$ and $P_{\rm error}$ respectively.

Furthermore, we should point out that $P_f=P_{\rm loss}+P_{\rm error}$ is a more significant criteria to evaluate IFM process, because it describes all the possibilities where the IFM process is a failure, including both the photon loss (object damage) and the error making in the discrimination process. However, even for this criteria $P_f$, we can also come to the conclusion that the quantum correlation (entanglement in our problem) can not benefit IFM process(see Appendix.~\ref{entanglement_enhance}). In addition, the asymptotic behaviors are also studied and we find that the state $\ket{1}$ in our system can perfectly detect the generic semitransparent object without any object damage when $N\rightarrow\infty$.

Finally, our paper provides principal theoretic support for the experimental research and practical realization of IFM, like electron microscopy of biological substances or detection of fragile nano-materials. Moreover, our theoretical approaches, borrowing from quantum information theory, such as quantum channel theory, quantum state discrimination etc, can be applied to other quantum facilitating scenarios and the analysis of whether quantum correlation can benefit these specific processes or not is intriguing.

\emph{Acknowledgement.}---We acknowledge C.~R.~Yang and X.~Yuan for the insightful discussions. This work was supported by the National Natural Science Foundation of China Grants No. 11405093.
\appendix
\section{derivation of the quantum channel $\mathcal{E}_I$ of the generic semitransparent object}\label{channel_derivation}
In the main part, the generic semitransparent object is composed of a beam splitter and a photon detector. The quantum channel $\mathcal{E}_I$ will be built by combing the operation of the beam splitter and the photon detector in the following.

let us give the channel description of the photon detector first.
The photon detector is modeled by a two-level atom with the ground state $\ket{g}$ and the exited state $\ket{e}$ respectively. And the atom staying at $\ket{g}$ interacts with the incident photon mode, denoted by $\ket{p}$. The atom can absorb the photon, transform it to the vacuum state $\ket{v}$ and become to the exited state $\ket{e}$ under the unitary $U_{\rm det}$; however, the unitary  $U_{\rm det}$ does not change the state $\ket{v,g}$, that is,
\begin{equation}\label{unitary_detector}
\begin{aligned}
&U_{\rm det}\ket{p,g}=\ket{v,e},\\
&U_{\rm det}\ket{v,g}=\ket{v,g}.
\end{aligned}
\end{equation}
Then the atom should be measured in the $\ket{g}, \ket{e}$ basis and reset to $\ket{g}$. In fact, we does not need to care about the operation of $U_{\rm det}$ on the other two states, say, $\ket{p,e}$ and  $\ket{v,e}$, since the atom always stays at the ground state $\ket{g}$ before the interaction.

Hence, the overall operation on the photon state is:
\begin{equation}\label{detector_operation}
\begin{aligned}
\rho_{\rm out}&=\sum_{i=g,e}\bra{i}U_{\rm det}(\rho_{\rm in}\otimes \ket{g}\bra{g})U_{\rm det}^{\dag} \ket{i},\\
&=\sum_{i=g,e}\bra{i}U_{\rm det}\ket{g}\rho_{\rm in}\bra{g}U_{\rm det}^{\dag} \ket{i},
\end{aligned}
\end{equation}
where $\rho_{\rm in}$ and $\rho_{\rm out}$ are the input and output photon state.
Following the standard quantum channel construction method \cite{nielsen2010quantum}, the quantum operation in Eq.~\eqref{detector_operation} can be written down with the Kraus representation as:
\begin{equation}\label{}
\begin{aligned}
&\rho_{\rm out}=\sum_{i=0,1}K_i \rho_{\rm in} K_i^{\dag},\\
&K_0=\bra{g}U_{det}\ket{g}=\ket{v}\bra{v},\\
&K_1=\bra{e}U_{det}\ket{g}=\ket{v}\bra{p}.\\
\end{aligned}
\end{equation}
where $K_0$, $K_1$ are the corresponding Kraus operators and we obtain the expressions of them using Eq.~\eqref{unitary_detector}.

For the scenario in the main part, there are three photon modes, i.e., $\ket{1},\ket{2},\ket{3}$, except the vacuum one $\ket{v}$, and only $\ket{3}$ can interact with the detector. So the channel should be slightly modified to
\begin{equation}\label{}
\begin{aligned}
&\mathcal{E}_{det}(\cdot)=\sum_{i=0,1}D_i(\cdot)D_i^\dag,\\
&D_0=\ket{1}\bra{1}+\ket{2}\bra{2}+\ket{v}\bra{v},\\
&D_1=\ket{v}\bra{3}.\\
\end{aligned}
\end{equation}
where $D_0$, $D_1$ are the corresponding Kraus operators.

On the other hand, the matrix representation of the unitary for the beam splitter $U_{\rm b}$ in the $\ket{1},\ket{2},\ket{3},\ket{v}$ basis shows
\begin{equation}
  U_{\rm b}=\left(
  \begin{array}{cccc}
    1 & 0 & 0 & 0\\
    0 & a & -\sqrt{1-a^2} & 0\\
    0 & \sqrt{1-a^2}& a & 0\\
    0 & 0 & 0 & 1\\
  \end{array}
\right)
\end{equation}

Combing the two operations of the photon detector $\mathcal{E}_{\rm det}$ and the beam splitter $U_b$, we have the combined channel being
\begin{equation}\label{}
\begin{aligned}
&\mathcal{E}_{\rm com}(\cdot)=\sum_{i=0,1}D_iU_b(\cdot)U_b^\dag D_i^\dag,\\
\end{aligned}
\end{equation}
and the corresponding Kraus operators $C_i=D_iU_b$ show
\begin{equation}\label{channel_123v}
\begin{aligned}
&C_0=\ket{1}\bra{1}+a\ket{2}\bra{2}-\sqrt{1-a^2}\ket{2}\bra{3}+\ket{v}\bra{v},\\
&C_1=\sqrt{1-a^2}\ket{v}\bra{2}+a\ket{v}\bra{3}.\\
\end{aligned}
\end{equation}

In fact, the component $\ket{3}$ is redundant as it is introduced to illustrate the intermediate process between the beam splitter and the photon detector. Remembering that the beam splitter and the photon detector as a whole represent the semitransparent object, thus we can treat them together as a black box and the photon state in IFM equivalently lives in the three dimensional space $\mathcal{H}_{12v}=spanned\{\ket{1},\ket{2},\ket{v}\}$. As a result, without altering the function of the channel that represents the semitransparent object, we can eliminate the terms in the above Kraus operator (Eq.~\eqref{channel_123v}) that relate to the component $\ket{3}$ and get
\begin{equation}\label{channel_12v}
\begin{aligned}
&A_0=\ket{1}\bra{1}+a\ket{2}\bra{2}+\ket{v}\bra{v},\\
&A_1=\sqrt{1-a^2}\ket{v}\bra{2},\\
\end{aligned}
\end{equation}
where we use $A_0$ and $A_1$ to denote the new Kraus operators.

Further more, actually, we can substitute the loss state $\ket{3}$ for the vacuum state $\ket{v}$ in the above Kraus operators to obtain the effective channel in the main part (Eq.~\eqref{channel_123}), since using which label to count the photon loss probability is equivalent here. The physical insight behind this is that the component $\ket{3}$ reflected by the beam splitter should be absorbed totally by the photon detector, i.e., dephased and transformed to the vacuum state $\ket{v}$.

\section{non-increasing of the generalized trace distance under quantum operation}\label{contractive}
Here, we first give the definition of the generalized trace distance as follows.
\begin{definition}
The generalized trace distance for the two quantum state $\rho_1$ and $\rho_2$ shows,
\begin{equation}
D_q(\rho_1,\rho_2)=\|q\rho_1-(1-q)\rho_2\|,
\end{equation}
where $\|\cdots\|$ is the trace norm and $0\leq q\leq 1$ is the corresponding probability factor.
\end{definition}
Note that $D_{1/2}(\rho_1,\rho_2)$ is the original trace distance \cite{nielsen2010quantum}. Then we show the property of the generalized trace distance in the following Theorem.
\begin{theorem}\label{general_trdistance}
suppose $\Lambda(\cdot)$ is a trace preserving quantum operation, then it is contradictive for the generalized trace distance, i.e.,
\begin{equation}
D_q(\rho_1,\rho_2)\geq D_q(\Lambda(\rho_1),\Lambda(\rho_2)).
\end{equation}
\end{theorem}
To prove Th.~\ref{general_trdistance} conveniently, we show another equivalent definition for the generalized trace distance $D_q(\rho_1,\rho_2)$.
\begin{lemma}\label{alter_def_trnorm}
\begin{equation}
D_q(\rho_1,\rho_2)=\mathrm{Tr}_{max}[(P_1-P_0)M],
\end{equation}
where we use $M=q\rho_1-(1-q)\rho_2$ for simplicity; and the maximization is over all projector pairs $P_0$, $P_1$ that satisfy $P_0+P_1=\mathbb{I}$.
\end{lemma}

\begin{proof}
$M$ is a hermit matrix by definition, thus we can use unitary to diagonalize it to $UMU^\dag$, and by separating the eigenvalues to nonnegative and negative parts we can obtain $UMU^\dag=Q'-S'$. As a result, we can represent $M$ as the subtraction of the two nonnegative matrices $M=U^\dag(Q'-S')U=Q-S$, and $||M||=||Q-S||=Tr(Q)+Tr(S)$. Then for any projector pair $P_0$, $P_1$,
\begin{equation}
\begin{aligned}
  \mathrm{Tr}[(P_1-P_0)M]&=Tr[(P_1-P_0)(Q-S)],\\
 &\leq \mathrm{Tr}[P_1Q+P_0S],\\
 &\leq \mathrm{Tr}(Q)+\mathrm{Tr}(S),\\
 &\leq ||M||.
\end{aligned}
\end{equation}
We can choose $P_0$ and $P_1$ just the projectors on the two orthogonal subspace where $Q$ and $S$ lives respectively, then $Tr[(\Pi_1-\Pi_0)M]$ can reach $||M||$ in this way and we finish the proof.
\end{proof}
Then we prove Th.~\ref{general_trdistance} with the help of Lemma.~\ref{alter_def_trnorm}
\begin{proof}
\begin{equation}
\begin{aligned}
\|M\|&=Tr(Q)+Tr(S),\\
 &= Tr[\Lambda(Q)+\Lambda (S)],\\
 &\geq Tr[(P'_1-P'_0)(\Lambda (Q)-\Lambda (S))],\\
 &=Tr[(P'_1-P'_0)\Lambda(M)],\\
 &=\|\Lambda(M)\|,
\end{aligned}
\end{equation}
where $P'_0$ and $P'_1$ are the projector pair used to reach the maixmal value $\|\Lambda(M)\|$, referring to Lemma.~\ref{alter_def_trnorm}.
Then, by substituting $M=q\rho_1-(1-q)\rho_2$, we finish the proof.
\end{proof}

\section{proof of Lemma.~\ref{tr_fidelity_lemma}}\label{app_lemma_tr_fidelity}
Here, we give the proof of Lemma.~\ref{tr_fidelity_lemma} in the main part that says
\begin{equation*}
  \|p|\psi_1\rangle\langle \psi_1| - |\psi_2\rangle \langle \psi_2| \|=\sqrt{(p+1)^2-4p|\langle \psi_1|\psi_2\rangle|^2}.
 \end{equation*}
\begin{proof}
$|\psi_1\rangle\langle \psi_1|$ can be expressed as $\frac{1}{2}(I+\sigma_z)$ in the basis of itself. Since $|\psi_2\rangle\langle \psi_2|$ does not change if we change the global phase of it, we have $|\psi_2\rangle=\cos\frac{\gamma}{2}|\psi_1\rangle+\sin\frac{\gamma}{2}|\psi_3\rangle$ $(0\leq\gamma \leq \pi/2)$, where $|\psi_3\rangle$ is the state orthogonal to $|\psi_1\rangle$. Then $|\psi_2\rangle\langle \psi_2|$ shows $\frac{1}{2}(I+\cos \gamma\sigma_z+\sin\gamma \sigma_x)$. And the trace norm $\|p|\psi_1\rangle\langle \psi_1| - |\psi_2\rangle \langle \psi_2| \|=\frac{1}{2}\|(p-1)I+(p-\cos \gamma)\sigma_z-\sin \gamma\sigma_x\|$. The two eigenvalues of $(p-1)\mathbb{I}+(p-\cos \gamma)\sigma_z-\sin \gamma\sigma_x$ are $(p-1)\pm\sqrt{(p-\cos\gamma)^2+\sin ^2 \gamma}$. Hence $\|p|\psi_1\rangle\langle \psi_1| - |\psi_2\rangle \langle \psi_2| \|=\sqrt{(p-\cos \gamma)^2+\sin ^2 \gamma}=\sqrt{(p+1)^2-4p\cos^2\frac{\gamma}{2}}=\sqrt{(p+1)^2-4p|\langle \psi_1|\psi_2\rangle|^2}$
\end{proof}
\section{The effect of quantum correlation for IFM process considering $P_{\rm error}$ and $P_f$ }\label{entanglement_enhance}
In this appendix, we will give the detailed illustration of the argument in the main part of our article that quantum correlation can not benefit the IFM process considering $P_{\rm error}$ and $P_f$ respectively.

For simplicity, we substitute for the terms in Eq.~\eqref{perror_general} by:
\begin{equation}\label{lambda_12}
\begin{aligned}
\lambda_1&=qTr[C^\dag C\rho_A]+(1-q),\\
\lambda_2&=2\sqrt {q(1-q)}|Tr[D^\dag C\rho_A]|.\\
\end{aligned}
\end{equation}
As a result, Eq.~\eqref{perror_general} becomes to a more concise form,
\begin{equation}\label{Perror_lambda_12}
\begin{aligned}
P_{\rm error}=\frac{1}{2}(\lambda_1-\sqrt{\lambda_1^2-\lambda_2^2}).
\end{aligned}
\end{equation}
It is not difficult to see that $P_{\rm error}$  monotonically decreases with the increasing of $\lambda_1$ due to the first order partial derivative on $\lambda_1$ being
\begin{equation}\label{}
\begin{aligned}
\frac{\partial P_{\rm error}}{\partial \lambda_1}=\frac{1}{2}(1-\frac{\lambda_1} {\sqrt{\lambda_1^2-\lambda_2^2}})
\leq 0.
\end{aligned}
\end{equation}
In the meantime, it's obvious that $P_{\rm error}$ decreases with the decreasing of $\lambda_2$.  Consequently, increasing $\lambda_1$ and decreasing $\lambda_2$ at the same time can minimize $P_{\rm error}$.

Here, with the help of Eqs.~\eqref{CdagC},~\eqref{DdagC} and~\eqref{two_level_state}, we present the expressions for $Tr[C^\dag C\rho_A]$, $Tr[D^\dag C\rho_A]$ in the definitions of $\lambda_1$ and $\lambda_2$ explicitly as
\begin{equation}\label{app_traceeq}
\begin{aligned}
Tr[C^\dag C\rho_A]&=(\frac{1-a}{2})^{2N}[f_1^2(1+k_1^2)+f_2^2+2f_1(f_2r_z-f_1k_1r_x)],\\
Tr[D^\dag C\rho_A]&=(\frac{1-a}{2})^{N}[f_1k_1-f_1r_x+i f_2r_y].
\end{aligned}
\end{equation}
The above equations shows that for a fixed $r_x$, we can always increase $\lambda_1$ and decrease $\lambda_2$ (i.e., decrease $P_{\rm error}$) by changing $r_y=0$ and $r_z=\sqrt{1-r_x^2}$. In other words, the minimum of $P_{\rm error}$ should be reached by pure state $\rho_A=\ket{\varphi_A}\bra{\varphi_A}$ of photon $A$ part. Thus, entangled photon input state can not enhance the performance for $P_{\rm error}$ in all $k_1$ regime, compared with single photon input state.

Moreover, we consider the effect of quantum correlation to IFM process, with a more effective criteria $P_f=P_{\rm loss}+P_{\rm error}$, which describing all the failure probability in IFM process. Owing to Th.~\ref{t1}, $P_f$ also shows the following concave property like $P_{\rm error}$,
\begin{equation}\label{ae5}
\begin{aligned}
P_f\geq p_i P_f^i.
\end{aligned}
\end{equation}
That is to say, the minimum of $P_f$ should be reached by the pure state and the quantum correlation here means entanglement. With the help of Eqs.~\eqref{Ploss_newbasis},~\eqref{perror_general} and~\eqref{lambda_12} , we have $P_f$ being
\begin{equation}\label{Pf_lamda_12}
\begin{aligned}
P_f=1-\frac{1}{2}(\lambda_1+\sqrt{\lambda_1^2-\lambda_2^2}).
\end{aligned}
\end{equation}
$P_f$ also decreases with the increasing of $\lambda_1$ and decreasing of $\lambda_2$. Consequently, just like the aforementioned reason for $P_{\rm error}$, we can also argue that entanglement can not enhance the performance of IFM considering $P_f$.
\section{the derivation of Eq.~\eqref{Plossmin_asym}}\label{Plossmin_asym_App}
Here, we give the derivation of $(P_{\rm loss})_{min}^{N\rightarrow\infty}$ in Eq.~\eqref{Plossmin_asym}, which is the asymptotic behaviour of $Eq.~\eqref{Plossmin_newbasis_semi}$ as $N\rightarrow0$.

First,
utilizing the same approximation technique used in $Eq.~\eqref{approximate_1}$ in the main part, we get
\begin{equation}\label{app_asymptotic}
\begin{aligned}
&(\frac{1-a}{2})^{N}(\Sigma_2-\Sigma_1)= O(a^N)\rightarrow0,\\
&(\frac{1-a}{2})^{N}\Sigma_1\simeq(\frac{1-a}{2})^{N}\Sigma_2=\frac{1}{2}-\frac{1+a}{1-a}\frac{\pi^2}{16N}+O(\frac{1}{N^3}).\\
\end{aligned}
\end{equation}

Then, let us consider the asymptotic behavior of $Eq.~\eqref{Plossmin_newbasis_semi}$, the minimum of $P_{\rm loss}$. By applying  $Eq.~\eqref{sum_k1small},~\eqref{f_k1small}$ and the definitions of $k_1$, $k_2$ (Eq.~\eqref{k1_k2}), we have
\begin{equation}\label{}
\begin{aligned}
&(\frac{1-a}{2})^{2N}(f_1+\sqrt{f_2^2+f_1^2k_1^2})^2,\\
\simeq&(\frac{1-a}{2})^{2N}(f_1+f_2+\frac{f_1^2}{2f_2}k_1^2)^2,\\
=&(\frac{1-a}{2})^{2N}(\frac{\Sigma_1}{\sqrt{1-k_1^2}}+\Sigma_2+\frac{\Sigma_1^2}{2\Sigma_2}
\frac{k_1^2}{1-k_1^2})^2,\\
\simeq&(\frac{1-a}{2})^{2N}[\Sigma_1(1+\frac{k_1^2}{2})+\Sigma_2+\frac{\Sigma_1^2}{2\Sigma_2}
k_1^2(1+k_1^2)]^2,\\
\simeq&(\frac{1-a}{2})^{2N}[(\Sigma_1+\Sigma_2)+\frac{k_1^2}{2}\Sigma_1(1+\frac{\Sigma_1}
{\Sigma_2})]^2,\\
\simeq&[1-\frac{1+a}{1-a}\frac{\pi^2}{8N}+(\frac{1+a}{1-a})^2\frac{\pi^2}{8N^2}+O(\frac{1}{N^3})]^2,\\
\simeq&1-\frac{1+a}{1-a}\frac{\pi^2}{4N}+O(\frac{1}{N^2}).
\end{aligned}
\end{equation}
where in the next-to-last row, we employ the equalities in Eq.~\eqref{app_asymptotic}.
So the asymptotic expression of $Eq.~\eqref{Plossmin_newbasis_semi}$ is
\begin{equation*}
\begin{aligned}
(P_{\rm loss})_{min}^{N\rightarrow\infty}\simeq q[\frac{1+a}{1-a}\frac{\pi^2}{4N}-O(\frac{1}{N^2})].
\end{aligned}
\end{equation*}
just as Eq.~\eqref{Plossmin_asym} describes in the main part.

\bibliographystyle{apsrev4-1}

\bibliography{zyzybib}

\end{document}